\documentclass[a4paper,UKenglish,autoref, thm-restate]{lipics-v2021}
\usepackage[nameinlink]{cleveref}
\usepackage{graphicx}

\usepackage[subrefformat=simple,labelformat=simple]{subcaption}

\crefname{figure}{Figure}{Figures}
\crefname{theorem}{Theorem}{Theorems}
\crefname{lemma}{Lemma}{Lemmas}
\crefname{corollary}{Corollary}{Corollaries}
\crefname{section}{Section}{Sections}
\crefname{observation}{Observation}{Observations}
\crefname{appendix}{Appendix}{Appendices}
\crefname{claim}{Claim}{Claims}

\usepackage{todonotes}
\usepackage{siunitx}

\newcommand\sbullet[1][.5]{\mathbin{\vcenter{\hbox{\scalebox{#1}{$\bullet$}}}}}

\newtheorem{problem}[theorem]{Problem}

\title{Guarding Polyominoes Under $k$-Hop Visibility}

\author{Omrit Filtser}{Department of Mathematics and Computer Science, The Open University of Israel,  Israel \and \url{https://omrit.filtser.com/}}{omrit.filtser@gmail.com}{https://orcid.org/0000-0002-3978-1428}{}

\author{Erik Krohn}{Department of Computer Science, University of Wisconsin - Oshkosh, Oshkosh, WI, USA}{krohne@uwosh.edu}{https://orcid.org/0000-0002-5832-8135}{}

\author{Bengt~J. {Nilsson}}{Department of Computer Science and Media Technology, Malm\"o University, Sweden \and \url{http://webshare.mah.se/tsbeni/} }{bengt.nilsson.TS@mau.se}{https://orcid.org/0000-0002-1342-8618}{}

\author{Christian Rieck}{Department of Computer Science, TU Braunschweig, Germany}{rieck@ibr.cs.tu-bs.de}{https://orcid.org/0000-0003-0846-5163}{}

\author{Christiane Schmidt}{Department of Science and Technology, Link\"oping University, Sweden \and \url{https://www.itn.liu.se/~chrsc91/}}{christiane.schmidt@liu.se}{https://orcid.org/0000-0003-2548-5756}{}

\authorrunning{O. Filtser, E. Krohn, B.~J. Nilsson, C. Rieck, and C. Schmidt}

\Copyright{Omrit Filtser and Erik Krohn and Bengt~J. Nilsson and Christian Rieck and Christiane Schmidt} 


\category{}

\relatedversion{}

\funding{O.\,F.\ was partially supported by the Israel Science Foundation (grant no. 2135/24).
	B.\,J.\,N.\ and C.\,S.\ are supported by the grant ``{\sc illuminate} provably good methods for guarding problems'' (2021-03810) and were supported by the grant ``New paradigms for unmanned air traffic management'' (2018-04001) from the Swedish Research Council.
	C.\,S.\ was supported by grant 2018-04101 (NEtworK Optimization for CarSharing Integration into a Multimodal TRANSportation System) from Sweden’s innovation agency VINNOVA.}

\nolinenumbers

\hideLIPIcs

\EventEditors{} 
\EventNoEds{0}
\EventLongTitle{}
\EventShortTitle{}
\EventAcronym{}
\EventYear{}
\EventDate{}
\EventLocation{}
\EventLogo{}
\SeriesVolume{}
\ArticleNo{}

\usepackage{multicol}
\usepackage{complexity}
\usepackage{xspace}
\usepackage{nicefrac}
\newcommand{\planarsat}{\textsc{Planar Monotone 3Sat}\xspace}
\def\eps{{\varepsilon}}

\newcommand{\rb}{\operatorname{rb}}

\begin{document}

\maketitle
	
\begin{abstract}
We study the \textsc{Art Gallery Problem} under $k$-hop visibility in polyominoes. 
In this visibility model, two unit squares of a polyomino can see each other if and only if the shortest path between the respective vertices in the dual graph of the polyomino has length at most~$k$.

In this paper, we show that the VC dimension of this problem is \num{3} in simple polyominoes, and \num{4} in polyominoes with holes. 
Furthermore, we provide a reduction from \planarsat, thereby showing that the problem is \NP-complete even in thin polyominoes (i.e., polyominoes that do not contain a $2\times 2$ block of cells).
Complementarily, we present a linear-time $4$-approximation algorithm for simple $2$-thin polyominoes (which do not contain a $3\times 3$ block of cells) for all $k\in \mathbb{N}$.
\end{abstract}

\section{Introduction}\label{sec:intro}
``How many guards are necessary and sufficient to guard an art gallery?'' 
This question was posed by Victor Klee in 1973, and led to the classic \textsc{Art Gallery Problem}: 
Given a polygon $P$ and an integer $\ell$, decide whether there is a guard set of cardinality~$\ell$ such that every point $p\in P$ is seen by at least one guard, where a point is seen by a guard if and only if the connecting line segment is inside the polygon. 

Now picture the following situation: A station-based transportation service (e.g., carsharing) wants to optimize the placement of their service stations. 
Assume that the demand is given in a granularity of (square) cells, and that customers are willing to walk a certain distance (independent of where they are in the city) to a station. 
Then, we aim to place as few stations as possible to serve an entire city for a given maximum walking range of $k$ cells. We thus represent the city as a polyomino, potentially with holes, and only walking within the boundary is possible (e.g., holes would represent water bodies or houses, which pedestrians cannot cross).

Our real-world example therefore motivates the following type of visibility. 
Two cells $u$ and $v$ of a polyomino $P$ can see each other if the shortest geodesic path from $u$ to $v$ in the polyomino has length at most $k$. 
For $k\geq 2$, this allows for the situation that a cell sees other cells that are around or behind corners of the polyomino, as visualized in~\cref{fig:vis}.

\begin{figure}[htb]
	\centering
	\includegraphics[scale=0.6]{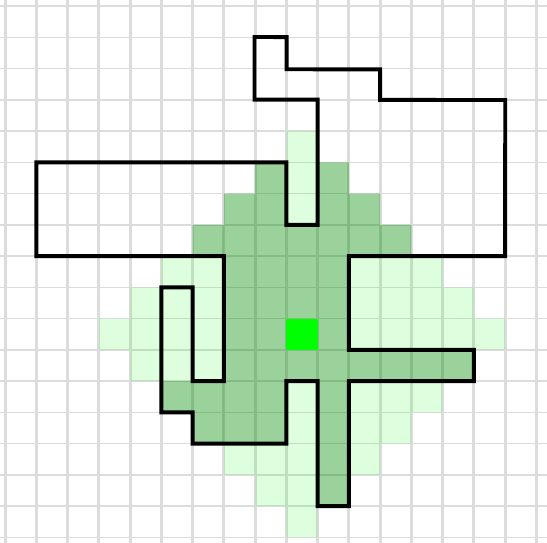}
	\caption{A unit square in green with its $k$-hop-visibility region for $k=6$ (shaded in dark green) within a polyomino---a subset of the diamond shown in light green.}
	\label{fig:vis}
\end{figure}

In this paper, we investigate the following problem.

\begin{problem}[\textsc{Minimum $k$-Hop Guarding Problem in Polyominoes} (M$k$GP)]
	Given a polyomino~$P$ and an integer $k$, find a minimum-cardinality unit-square guard cover in $P$ under $k$-hop visibility.
\end{problem}

As the dual graph of a polyomino is a grid graph, we analogously state the following. 

\begin{problem}[\textsc{Minimum $k$-Hop Dominating Set Problem in Grid Graphs} (M$k$DSP)]
	Given a grid graph $G$ and an integer $k$, find a minimum-cardinality subset $D_k\subseteq V(G)$, such that for any vertex $v\in V(G)$ there exists a vertex $u\in D_k$ within hop distance of at most $k$.
\end{problem}

\medskip
While we formulated the optimization variants, the associated decision problems are defined as expected with an upper bound on the number of guards or dominating vertices, respectively.

\paragraph*{Our Contributions}
In this paper, we give the following results for the \textsc{Minimum $k$-Hop Guarding Problem in Polyominoes}.
\begin{itemize}
	\item In~\cref{sec:vc}, we analyze the VC dimension of the problem and give tight bounds. 
	In~particular, we prove that inside a simple polyomino exactly \num{3} squares can be shattered by $k$-hop visibility, see~\cref{th:vc-sim}. 
	For polyominoes with holes, we show that the VC dimension of $k$-hop visibility is \num{4}, see~\cref{th:vc-hol}.
	\item  In~\cref{sec:np}, we study the computational complexity of the respective decision version of the problem.
	We show that the problem is \NP-complete for $k\geq 2$, even in $1$-thin polyominoes with holes (polyominoes that do not contain a $2\times 2$ block of unit squares), see \cref{th:np}.
	\item In~\cref{sec:appx}, we provide a linear-time $4$-approximation for simple polyominoes that do not contain a $3\times 3$ block of unit squares (i.e., $2$-thin polyominoes), see~\cref{th:4-appx}.
\end{itemize}

\paragraph*{Related Work} 
The classic \textsc{Art Gallery Problem} (AGP) is \NP-hard~\cite{ll-ccagp-86,os-snpdp-83}, even in the most basic problem variants. 
Abrahamsen, Adamaszek, and Miltzow~\cite{AbrahamsenAM22} recently showed that the AGP is $\exists\mathbb{R}$-complete, even when the corners of the given polygon are at integer~coordinates.

Guarding polyominoes and thin (orthogonal) polygons has been considered for different definitions of visibility. 
Tomás~\cite{t-gtoph-13} showed that computing a minimum guard set under the original definition of visibility is \NP-hard for point guards and \APX-hard for vertex or boundary guards in thin orthogonal polygons; an orthogonal polygon is defined as {\it thin} if the dual graph of the partition obtained by extending all polygon edges through incident reflex vertices is a tree.  
Biedl and Mehrabi~\cite{bm-orgto-16} considered guarding thin orthogonal polygons under rectilinear visibility (two points can see each other if the axis-aligned rectangle spanned by these points is fully contained in the polygon). 
They showed that the problem is \NP-hard in orthogonal polygons with holes, and provided a linear-time algorithm that computes a minimum set of guards under rectilinear vision for tree polygons. 
Their~approach generalizes to polygons with~$h$ holes or thickness~$t$ (the dual graph of the polygon does not contain an induced $(t+1)\times (t+1)$ grid); hence, the problem is fixed-parameter tractable in $t+h$. 
Biedl and Mehrabi~\cite{bm-ogopb-21} extended this study to orthogonal polygons with bounded treewidth under different visibility definitions usually used in orthogonal polygons: rectilinear visibility, staircase visibility (guards see along an axis-parallel staircase), and limited-turn path visibility (guards see along axis-parallel paths with at most $b$ bends). 
Under these visibility definitions, they showed the guarding problem to be  linear-time solvable. 
For~orthogonal polygons, Worman and Keil~\cite{wk-pdoag-07} gave a polynomial-time algorithm to compute a minimum guard cover under rectilinear visibility by exploiting that an underlying graph is~\emph{perfect}, and that this particular variant is equivalent to a clique cover problem.

Biedl et al.~\cite{BIIKM11} proved that determining the guard number of a given simple polyomino with $n$ unit squares is \NP-hard even in the all-or-nothing visibility model (a~unit square $s$ of the polyomino is visible from a guard $g$ if and only if $g$ sees all points of~$s$ under ordinary visibility), and under ordinary visibility. 
They presented polynomial-time algorithms for thin polyominoes, for which the dual is a tree, and for the all-or-nothing model with limited range of visibility. 
Iwamoto and Kume~\cite{ik-ccrvg-13} complemented the \NP-hardness results by showing \NP-hardness for polyominoes with holes also for rectilinear visibility.  Pinciu~\cite{p-gppp-15} generalized this to polycubes, and gave simpler proofs for known results and for guarding polyhypercubes.

The \textsc{Minimum $k$-Hop Dominating Set Problem} is \NP-complete in general graphs~\cite{apvh-mmdcf-00,bm-finsn-14}.
For trees, Kundu and Majunder~\cite{KM16} showed that the problem can be solved in linear time. 
Recently, Abu-Affash et al.~\cite{ack-ltamk-22} simplified that algorithm, and provided a linear-time algorithm for cactus graphs.
Borradaile and Le~\cite{BL16} presented an exact dynamic programming algorithm that runs in $O((2k+1)^{tw}\cdot n)$ time on graphs with treewidth~$tw$.
Demaine et al.~\cite{DFHT05} consider the decision version of the \textsc{Minimum $k$-Hop Dominating Set Problem}\footnote{In fact, Demaine et al. call the problem the  ($\ell,k$)-center problem.} on planar and map graphs, asking whether a graph can be covered with a dominating set of cardinality $\ell$.
They showed that for these graph families, the problem is fixed-parameter tractable by providing an exact $2^{O(k\log k)\sqrt{\textrm{OPT}}}\cdot\poly(n)$ time algorithm, where $\textrm{OPT}$ is the size of an optimal solution. 
They also obtained a $(1+\eps)$-approximation for these families that runs in $k^{O(\nicefrac{k}{\eps})}\cdot m$ time, where $m$ is the number of edges in the graph.

In the general case, where the edges of the graph are weighted, the problem is typically called the \textsc{$\rho$-Dominating Set Problem}. 
Katsikarelis et al.~\cite{klp-sptba-19} provided an \FPT\ approximation scheme parameterized by the graphs treewidth~$tw$, or its clique-width~$cw$. 
In~particular, if there exists a $\rho$-dominating set of size~$s$ in a given graph, the approximation scheme computes a $(1+\eps)\rho$-dominating set of size at most $s$ in time $(\nicefrac{tw}{\eps})^{O(tw)}\cdot \poly(n)$, or $(\nicefrac{cw}{\eps})^{O(cw)}\cdot \poly(n)$, respectively. 
Fox-Epstein et al.~\cite{fks-epglt-19} provided a bicriteria \EPTAS\ for $\rho$-domination in planar graphs (later improved and generalized by Filtser and Le~\cite{FL22}). 
Their algorithm runs in $O(n^c)$ time (for some constant $c$), and returns a $(1+\eps)\rho$-dominating set of size $s\le(1+\eps)\textrm{OPT}_{\rho}$, where $\textrm{OPT}_{\rho}$ is the size of a minimum $\rho$-dominating set. 
Filtser and Le~\cite{fl-cetlt-21} provided a \PTAS\ for $\rho$-domination in $H$-minor-free graphs, based on local~search. 
Their algorithm has a runtime of $n^{O_H(\eps^{-2})}$ and returns a $\rho$-dominating set of size at most $(1+\eps)\textrm{OPT}_{\rho}$. 
Meir and Moon~\cite{mm-rbpcnt-75} showed an upper bound of $\lfloor\frac{n}{k+1}\rfloor$ on the number of vertices in a $k$-hop dominating set of any tree with $n$ vertices. 
This bound trivially holds for general graphs by using any spanning~tree. 

\section{Notation and Preliminaries}\label{sec:not}
A \emph{polyomino} $P$ of \emph{size} $|P|=n$ is a connected polygon in the plane formed by joining together $n$ unit squares (also called \emph{cells}) on the square lattice.
The dual graph~$G_P$ of a polyomino has a vertex at the center point of each cell of~$P$, and there is an edge between two center points if their respective cells share an edge. 
Note that $G_P$ is a \emph{grid graph}, thus, any vertex has at most four neighbors. 
A~polyomino is \emph{simple} if it has no holes, that is, every inner face of its dual graph has unit area.
A polyomino $P$ is \emph{$t$-thin} if it does not contain a block of squares of size $(t+1)\times(t+1)$, and analogously, a~grid graph~$G$ is \emph{$t$-thin} if $P_G$ is $t$-thin, where $P_G$ denotes the polyomino (that is unique except for rotation) for which $G$ is the dual~graph. 

A~unit square $v\in P$ is \emph{$k$-hop-visible} to a unit square $u\in P$
if the shortest path from $u$ to $v$ in $G_P$ has length at most $k$. 
The \emph{$k$-hop-visibility region} of a unit square~$u\in P$ is the set of all unit squares that are $k$-hop-visible from $u$.
It is a subset of the diamond with diameter $2k$---the maximal $k$-hop-visibility region, see~\cref{fig:vis}.
A~\emph{witness set} is a set $W$ of unit squares, such that the $k$-hop-visibility regions of the elements in $W$ are pairwise disjoint. 
A~\emph{witness} placed at a unit square $u$ vouches that at least one \emph{guard} has to be placed in its $k$-hop-visibility region.

\medskip
As mentioned, Meir and Moon~\cite{mm-rbpcnt-75} showed an upper bound of $\lfloor\frac{n}{k+1}\rfloor$ for $k$-hop dominating sets for graphs with $n$ vertices.
To the best of our knowledge, no matching lower bound is provided, so we catch up by showing that there is, for every $k\in \mathbb{N}$, a simple $1$-thin grid graph that requires that many dominating vertices.
Stated in context of the guarding problem, we~show the following.

\begin{proposition}\label{le:nec}
	For every $k\in \mathbb{N}$, there exist simple polyominoes of size $n$ that require $\lfloor\frac{n}{k+1}\rfloor$ guards to cover their interior under $k$-hop visibility.
\end{proposition}
\begin{proof}
	We construct a double-comb like polyomino by alternately adding teeth of length $k$ to the top and the bottom of a row of unit squares (the handle).
	\Cref{fig:lb} depicts the construction for $k=1$ and $k=2$. 
	
	\begin{figure}[htb]
		\centering
		\begin{subfigure}[b]{0.5\textwidth}
			\centering
			\includegraphics[scale=0.6]{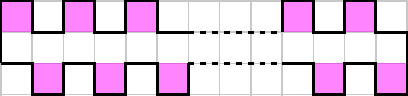}
			\caption{}
			\label{fig:lb-a}
		\end{subfigure}%
		\begin{subfigure}[b]{0.5\textwidth}
			\centering
			\includegraphics[scale=0.6]{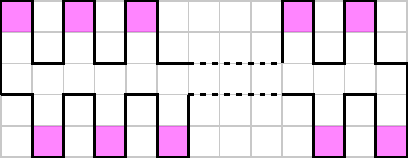}
			\caption{}
			\label{fig:lb-b}
		\end{subfigure}%
		\caption{\label{fig:lb} Lower bound construction for simple polyominoes that require $\lfloor\frac{n}{k+1}\rfloor$ guards under $k$-hop visibility for (a) $k=1$ and (b) $k=2$. Witnesses are shown in pink.}
	\end{figure}
	If $n$ is not divisible by $k+1$, we add $m=(n \bmod k+1)$ unit squares to the right of the handle. 
	Witnesses placed at the last unit square of each tooth (shown in pink) have disjoint $k$-hop-visibility regions (in particular, only the unit square of the handle to which the tooth is attached belongs to the $k$-hop-visibility region), hence, we need a separate guard for each witness. 
	The $m$~cells to the right of the handle can be covered by the rightmost guard if placed in the handle. 
	If $t$ denotes the number of teeth with $n=t\cdot(k+1)+m$, we clearly need $t=\lfloor\frac{n}{k+1}\rfloor$ guards.
\end{proof}

\section{VC Dimension}\label{sec:vc}
The \emph{VC dimension}, originally defined by Vapnik and Chervonenkis~\cite{vc-dimension-definition}, is a measure of complexity of a set system. 
In our setting, we say that a finite set of unit squares (guards) $D$ in a polyomino~$P$ is \emph{shattered} if for any of the $2^{|D|}$ subsets $D_i\subseteq D$ there exists a unit square $u\in P$, such that from $u$ every unit square in $D_i$ but no unit square in $D\setminus D_i$ is $k$-hop visible (or symmetrically: from every unit square in~$D_i$, the unit square $u\in P$ is $k$-hop visible, but $u$ is not $k$-hop visible from any unit square in $D\setminus D_i$). 
We then say that the unit square $u$ is a \emph{viewpoint}. 
The VC dimension is the largest $d$, such that there exists a polyomino $P$ and a set of $d$ unit-square guards $D$ that can be shattered. 
For detailed definitions, we refer to Haussler and Welzl~\cite{hw-ensrq-87}.

In this section, we study the VC dimension of the M$k$GP in both simple polyominoes and polyominoes with holes, and provide exact values for both~cases. The~VC~dimension has been studied for different guarding problems, e.g., Lange\-tepe and Lehmann~\cite{LL17} showed that the VC dimension of $L_1$-visibility in a simple polygon is exactly $5$, Gibson et al.~\cite{gkw-vcdvb-15} proved that the VC dimension of visibility on the boundary of a simple polygon is exactly $6$. 
For~line visibility in a simple polygon, the best lower bound of $6$ is due to Valtr~\cite{p-ggwnp-98}, the best upper bound of~$14$ stems from Gilbers and Klein~\cite{gk-nubvc-14}. 
Furthermore, given any set system with constant VC dimension, Brönnimann and Goodrich~\cite{bg-aoscf-95} presented a polynomial time $O(\log\textrm{OPT})$-approximation for \textsc{Set Cover}.

\smallskip
For analyzing the VC dimension, we define the \emph{rest budget} of a unit square $g\in P$ at a unit square~$c\in P$ to be $\rb(g,c)=\max\{k-d(g,c),0\}$, where $d(g,c)$ is the minimum distance between $g$ and $c$ in $G_P$, and $k$ the respective hop distance.
We~first state two structural properties which are helpful in several arguments.

\begin{observation}[Rest-Budget Observation]\label{le:rest-budget}
	Let $P$ be a polyomino, and let $g$ and $u$ be two unit squares in~$P$ such that a shortest path between them contains a unit square~$c$. Then the following holds:
	\begin{enumerate}
		\item The unit square $g$ covers $u$, if and only if $u$ is within distance $\rb(g,c)$ from~$c$.
		\item For any unit square $g'$ with $\rb(g',c)>\rb(g,c)$, if $g$ covers $u$, then so does~$g'$.
	\end{enumerate}
\end{observation}

\begin{lemma}[Rest-Budget Lemma]\label{lem:rest_budget_triangle}
	Let $a,c$ be two unit squares in a simple polyomino~$P$, such that the boundary of $P$ does not cross the line segment $\overline{ac}$ that connects their center points. Let $P_{ac}$ be some path in the dual graph $G_P$ between the center points of $a$ and $c$, and let $b$ be a unit square whose center point belongs to the area enclosed within $P_{ac}\circ \overline{ac}$. Then, there exists a unit square $x$ on $P_{ac}$ such that $\rb(b,x) \geq \rb(a,x)$ and $\rb(b,x) \geq \rb(c,x)$.
\end{lemma}
\begin{proof}
	Without loss of generality, assume that the center of $a$ is placed on the origin, $c$~lies in the first quadrant, and $b$ is above the line through the centers of $a$ and $c$; see~\cref{fig:rest-budget-2} for an illustrative example. 
	
	\begin{figure}[htb]
		\centering
		\includegraphics[scale=1]{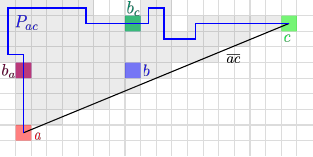}
		\caption{Location of unit squares used in the proof of~\cref{lem:rest_budget_triangle}.}
		\label{fig:rest-budget-2}
	\end{figure}
	
	If $b$ is above $c$, then let $x$ be the unit square on $P_{ac}$ directly above~$b$.
	As $P$ is simple, and the boundary of $P$ does not cross $\overline{ac}$, the area enclosed within $P_{ac}\circ \overline{ac}$ does not contain any boundary piece of $P$. Thus, the path in $G_P$ from $b$ to $x$ is a straight line segment, and we have $\rb(b,x) \geq \rb(a,x)$ and $\rb(b,x) \geq \rb(c,x)$, as required.
	Symmetrically, if $b$ is to the left of $a$, then let $x$ be the unit square on $P_{ac}$ directly to the left of $b$, and again we have $\rb(b,x) \geq \rb(a,x)$ and $\rb(b,x) \geq \rb(c,x)$, as required.
	
	The only case left is when $b$ lies in the axis-aligned bounding box of $a$ and $c$. 
	In this case, let $b_c$ be the unit square on $P_{ac}$ directly above $b$, and $b_a$ be the unit square on $P_{ac}$ directly to the left of $b$. We show that if the slope of $\overline{ac}$ is at most $1$ (i.e., $\frac{c_y}{c_x}\le 1$), then $\rb(b,b_c) \geq \rb(a,b_c)$ and $\rb(b,b_c) \geq \rb(c,b_c)$. The case when the slope is larger than $1$ is argued symmetrically, using the square $b_a$ instead of $b_c$.
	Denote the center point of~$b$ by $(b_x,b_y)$, and the center point of~$c$ by $(c_x,c_y)$. Let $y(b_c)$ be the $y$-coordinate of the center point of $b_c$, then $d(b,b_c)=y(b_c)-b_y$.
	It is easy to see that $\rb(b,b_c) \geq \rb(a,b_c)$, because $d(a,b_c)\ge b_x+y(b_c)>y(b_c)-b_y=d(b,b_c)$.
	To show that $\rb(b,b_c) \geq \rb(c,b_c)$, notice that since $c_y\le c_x$ and $b$ is above $\overline{ac}$, we have $c_x-b_x\ge c_y-b_y$.
	If $y(b_c)\ge c_y$ then $d(c,b_c)\ge c_x-b_x+y(b_c)-c_y\ge c_y-b_y+y(b_c)-c_y=y(b_c)-b_y=d(b,b_c)$. Else, $y(b_c)< c_y$, then $d(c,b_c)\ge c_x-b_x+c_y-y(b_c)\ge c_y-b_y+c_y-y(b_c)>y(b_c)-b_y=d(b,b_c)$.
\end{proof}

\subsection{Simple Polyominoes}\label{sec:vc-sim}
In this section, we investigate the VC dimension of $k$-hop visibility in simple polyominoes.
In particular, we show the following.

\begin{theorem}\label{th:vc-sim}
	For any $k\in \mathbb{N}$, the VC dimension of $k$-hop visibility in a simple polyomino is~$3$.
\end{theorem}

\begin{proof}
	A lower bound construction with $k=1$ and three guards, indicated by the green $1$, the blue $2$, and the red $3$, is visualized in~\cref{fig:vc-dim}.
	
	\begin{figure}[htb]
		\centering
		\includegraphics[scale=0.7]{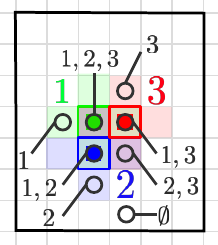}
		\caption{\label{fig:vc-dim}
			A lower bound construction for the VC dimension of $k$-hop visibility in simple polyominoes with $k=1$. 
			The  positions of the guards $1,2,3$ are shown as squares in green, blue, and red, respectively. 
			The $k$-hop-visibility regions are shown in a light shade of those colors. 
			The~\num{8}~viewpoints are indicated by circles, and labeled accordingly.
		}
	\end{figure}
	All $2^3=8$ viewpoints are highlighted and denoted by the guards that see each viewpoint.  For larger $k$,  we keep the placement of guards, but the polyomino will be a large rectangle that contains all $k$-hop-visibility regions.  Because of the relative position of the guards they are shattered as before.
	
	We now show that four guards cannot be shattered in simple polyominoes.
	To this end, consider four guards $g_1,g_2,g_3,g_4$ to be placed in the simple polyomino, and denote the potential viewpoints as $v_S$ with $S\subseteq\{1,2,3,4\}$. 
	For two unit squares $x,y\in P$, we denote by $sp(x,y)$ a shortest path between $x$ and $y$ in $G_P$. 
	For $i,j\in\{1,2,3,4\}$, let $P_{i,j}=sp(g_i,v_{\{i,j\}})\circ sp(v_{\{i,j\}},g_j)$.
	We now distinguish two cases depending on how many of the four guards lie on their convex hull.
	
	\subparagraph{Case 1: All four guards lie in convex position}
		That is, the four center points of their corresponding grid squares are in convex position. Pick any guard and label it~$g_1$ and then label the other in clockwise order around the convex hull $g_2, g_3$ and~$g_4$, see~\cref{fig:vc-simple-a-2}. 
		
		\begin{figure}[htb]
			\centering
			\begin{subfigure}[b]{.5\textwidth}%
				\centering
				\includegraphics[scale=0.75]{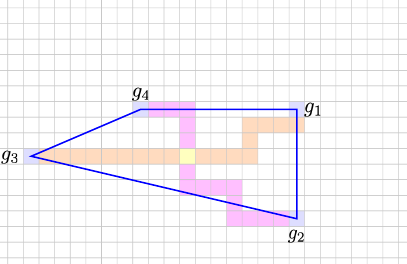}
				\caption{}
				\label{fig:vc-simple-a-2}
			\end{subfigure}\hfill
			\begin{subfigure}[b]{.5\textwidth}%
				\centering
				\includegraphics[scale=0.75]{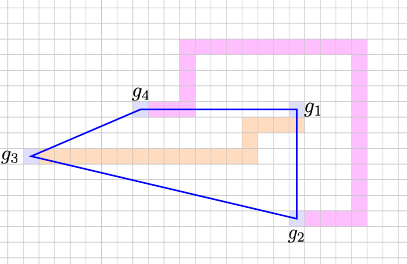}
				\caption{}
				\label{fig:vc-simple-b-2}
			\end{subfigure}\hfill\vspace{0.2cm}
			\begin{subfigure}[b]{.5\textwidth}%
				\centering
				\includegraphics[scale=0.75]{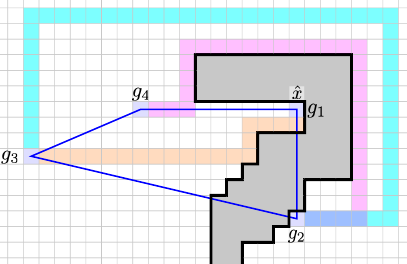}
				\caption{}
				\label{fig:vc-simple-c-2}
			\end{subfigure}\hfill
			\caption{\label{fig:vc-simple-2} Case 1 in the proof of~\cref{th:vc-sim},  with $P_{1,3}, P_{2,4}$ and $P_{2,3}$ shown in orange, pink and turquoise, respectively. (a) $P_{1,3}$ and $P_{2,4}$ cannot cross. (b) $P_{2,4}$ goes around~$g_1$. (c) The boundary of $P$ must pierce $\overline{g_2g_4}$ and both $P_{2,4}$ and $P_{2,3}$ must go around a square $\hat{x}$ in the exterior of $P$.}
		\end{figure}
		Assume, without loss of generality, that $g_4$ is to the left of $g_1$ and that~$g_1$ is above the line through $g_2$ and $g_4$.
		First, we claim that the paths $sp(g_i,v_{\{1,3\}})$ and $sp(g_j,v_{\{2,4\}})$ cannot cross each other for $i\in\{1,3\},j\in\{2,4\}$. Indeed, if the paths have a unit square $x$ in common, as depicted in~\cref{fig:vc-simple-a-2}, then one of $g_i$, $g_j$ has a larger rest budget at $x$ (or an equal rest budget). 
		Assume, without loss of generality, that $\rb(g_i,x)\ge \rb(g_j,x)$, then~$g_i$ would also cover $v_{\{2,4\}}$, which is a contradiction as $i\notin \{2,4\}$. Therefore, the paths $P_{1,3}=sp(g_1,v_{\{1,3\}})\circ sp(v_{\{1,3\}},g_3)$ and $P_{2,4}=sp(g_2,v_{\{2,4\}})\circ sp(v_{\{2,4\}},g_4)$ cannot cross, and one of them must ``go around'' a guard in order to avoid a crossing. Without loss of generality, assume that $P_{2,4}$ goes around $g_1$, that is, $g_1$ belongs to the area enclosed within $P_{2,4}\circ \overline{g_2g_4}$; see~\cref{fig:vc-simple-b-2}.
		Assume that the boundary of $P$ does not pierce~$\overline{g_2g_4}$. In this case, as $P$ is simple, we get by \Cref{lem:rest_budget_triangle} that there exists a square~$x$ on $P_{2,4}$ such that $\rb(g_1,x)\ge \rb(g_2,x)$ and $\rb(g_1,x)\ge \rb(g_4,x)$. Hence, $g_1$ covers $v_{\{2,4\}}$, a contradiction.
		
		We therefore assume that the boundary of $P$ does pierce $\overline{g_2g_4}$, as exemplified in~\cref{fig:vc-simple-c-2}, and, hence, there exists a square $\hat{x}\notin P$, which blocks $g_1$ from reaching the square $x$ on $P_{2,4}$ from \Cref{lem:rest_budget_triangle}. As $P$ is simple, the boundary of $P$ must also cross either $\overline{g_2g_3}$ or $\overline{g_3g_4}$ in a way that any path in $P$ between the endpoints of this segment must go around $\hat{x}$. In other words, assume, without loss of generality, that the boundary of $P$ crosses $\overline{g_2g_3}$. Then there exists a path in the exterior of~$P$ connecting $\overline{g_2g_3}$ and~$\hat{x}$, and because $P$ is simple, any path in $P$ from $g_2$ to $g_3$ must go around~$\hat{x}$. In~particular, the path $P_{2,3}=sp(g_2,v_{\{2,3\}})\circ sp(v_{\{2,3\}},g_3)$ also goes around~$\hat{x}$. 
		We get that both $P_{2,3}$ and $P_{2,4}$ go around $\hat{x}$; however, $sp(g_3,v_{\{2,3\}})$ and $sp(g_4,v_{\{2,4\}})$ cannot intersect. Moreover, consider the region $A^2_{3,4}$ that is enclosed by $\overline{g_3g_4}\circ sp(g_3,v_{\{2,3\}})\circ sp(v_{\{2,3\}},v_{\{2,4\}})\circ sp(g_4,v_{\{2,4\}})$, and assume that $sp(g_3,v_{\{2,3\}})$ is above $sp(g_4,v_{\{2,4\}})$ (the other case is argued analogously). As~$P$ is simple, the region $A^2_{3,4}$ does not contain any polyomino boundary. Consider the line $\ell$ through $g_4$ of slope~$-1$. If $g_3$ is below $\ell$, then for any unit square~$s$ to the right of $g_4$ inside $A^2_{3,4}$, we have $\rb(g_4,s)\ge \rb(g_3,s)$. As~$g_3$ also lies below $\overline{g_2g_4}$ and to the left of $g_4$ (and because $k$-hop-visibility regions are diamond-shaped when they do not intersect the boundary), we get that $g_4$ reaches $v_{\{2,3\}}$, a contradiction.
		On~the other hand, if $g_3$ is above $\ell$, consider the region $A^1_{3,4}$ enclosed by $\overline{g_3g_4}\circ sp(g_3,v_{\{1,3\}})\circ sp(v_{\{1,3\}},v_{\{1,4\}})\circ sp(g_4,v_{\{1,4\}})$ and assume that $sp(g_3,v_{\{1,3\}})$ is below $sp(g_4,v_{\{1,4\}})$. By the same arguments, the region~$A^1_{3,4}$ does not contain any polyomino boundary, and for any square $s$ below $g_4$ inside $A^1_{3,4}$, we have $\rb(g_4,s)\ge \rb(g_3,s)$. 
		In~this case, $g_4$ reaches $v_{\{1,3\}}$, a contradiction.
		
	\subparagraph{Case 2: Exactly three guards lie on their convex hull}
		That is, the three center points of their corresponding grid squares are in convex position, and the center point of the fourth guard lie in the convex hull. We label the three guards on the convex hull $g_1,g_2,g_3$, and the guard that is placed inside the convex hull is labeled $g_4$. 
		
		We show that the viewpoint~$v_{\{1,2,3\}}$ is not realizable.
		Let $T$ be the triangle of grid points that connects the centers of $g_1, g_2$,~and~$g_3$.
		Consider the three shortest paths connecting $g_1,g_2,g_3$ to $v_{\{1,2,3\}}$. As $g_4$ lies in~$T$, for any placement of $v_{\{1,2,3\}}$, we would get that for some $i,j\in\{1,2,3\}$, the area enclosed within $sp(g_i,v_{\{1,2,3\}})\circ sp(v_{\{1,2,3\}},g_j)\circ \overline{g_ig_j}$ contains the center point of $g_4$.
		If the boundary of $P$ does not pierce~$T$, then, similar to Case~1, we get by \Cref{lem:rest_budget_triangle} that $g_4$ reaches $v_{\{1,2,3\}}$, a~contradiction. Otherwise, assume that the convex hull of the three guards is pierced by the boundary. Then it is possible to realize the $v_{\{1,2,3\}}$ viewpoint.  However, similar to the argument in Case~1, the boundary will prevent the realization of a viewpoint of~$g_4$ and one of the other guards ($g_4$ taking the role of $g_1$ from Case~1 here).
\end{proof}

\subsection{Polyominoes with Holes}\label{sec:vc-hol}

Aronov et al.~\cite{adep-opdh-21} showed an upper bound of $4$ for the VC dimension in hypergraphs of pseudo disks.  And while, intuitively, one might suspect that $k$-hop-visibility regions of unit squares in polyominoes with holes are pseudo disks; that is not the case, as illustrated in~\cref{fig:vc-holes-lb-a}.  

\begin{figure}[htb]
			\centering
			\includegraphics[scale=0.7]{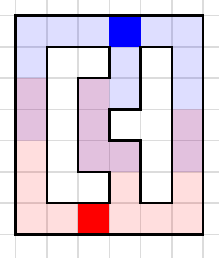}
			\caption{The $k$-hop-visibility regions ($k=6$) of two guards intersect more than twice.}
			\label{fig:vc-holes-lb-a}
\end{figure}%

Hence, we need to show an upper bound for the VC dimension in this case in another way. 
In fact, even here, we provide matching upper and lower bounds. 
These are valid for large enough values of $k$ (e.g., for $k=1$ we do not gain anything from the holes). 
In particular, we show the following. 

\begin{theorem}\label{th:vc-hol}
	For large enough $k\in \mathbb{N}$, the VC dimension of $k$-hop visibility in a polyomino with holes is $4$.
\end{theorem}

\begin{proof}
	A lower bound with $k=18$ is visualized in~\cref{fig:vc-holes-lb-b}: the four guards are indicated by the green $1$, the blue $2$, the red $3$, and the yellow $4$. 
	\begin{figure}[htb]
			\begin{subfigure}[b]{.5\textwidth}
					\centering
					\includegraphics[scale=0.5]{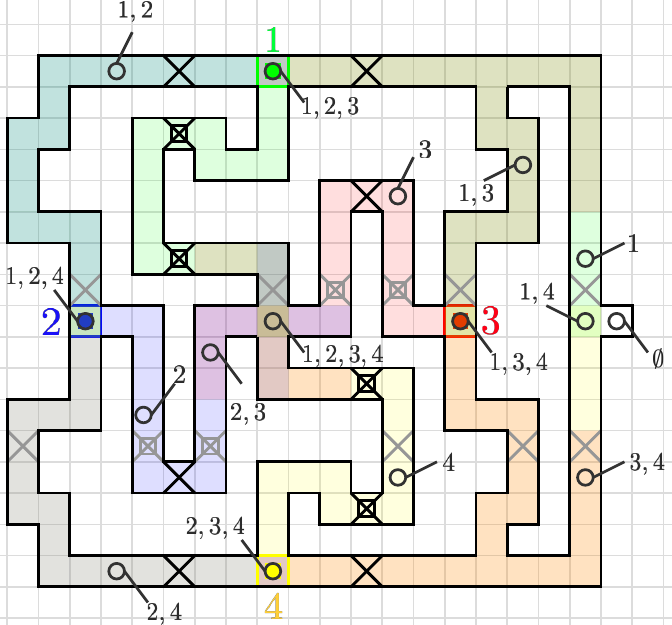}
					\caption{}
					\label{fig:vc-holes-lb-b}
				\end{subfigure}
				\begin{subfigure}[b]{.5\textwidth}
					\centering
					\includegraphics[scale=0.5]{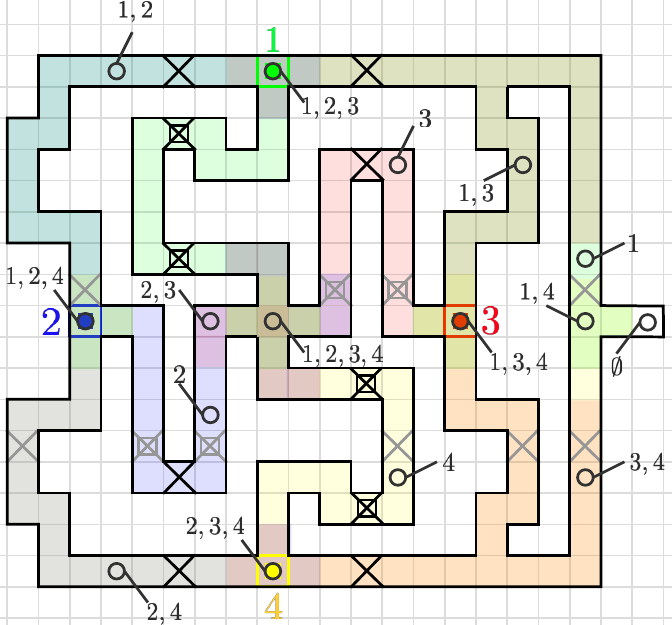}
					\caption{}
					\label{fig:vc-holes-lb-c}
			\end{subfigure}%
			\caption{(a) The lower bound construction for the VC dimension of $k$-hop visibility in non-simple polyominoes for $k=18$.  
						The positions of the four guards $1,2,3,4$ are shown as squares in green, blue, red, and yellow, respectively. 
						Visibility regions are shown in a light shade of those colors. 
						The \num{16}~viewpoints are indicated by circles, and labeled accordingly. The gray and black $\times$'s and boxed $\times$'s indicate where we insert $2$~and $1$ unit squares, respectively, to increase the value of $k$ by~$2$. We alternate between using the gray and black markings.  (b) The lower bound construction for the VC~dimension of $k$-hop visibility in non-simple polyominoes for $k=19$.  The colors are used as in (a).}
			\label{fig:lower-bound-vc-holes}
	\end{figure}
	We highlighted the $2^4=16$ viewpoints, and denoted them with the respective guard labels that see them.  For larger (even) values of $k$, we extend the corridors in~\cref{fig:vc-holes-lb-b} at the location marked by~``$\times$'': We alternate between using the gray and black unit squares.  At locations with a simple~``$\times$'', we insert two unit squares, at locations with a boxed ``$\times$'', we insert a single unit square. 
	One can verify that by alternating between the gray and black insertions for $k=18+2i$,  all viewpoints are realized.
	Similarly,  a lower bound with $k=19$ is shown in~\cref{fig:vc-holes-lb-c}.  Again we use the squares marked by~``$\times$'' for inserting one or two unit squares to extend corridors for $k=19+2i$.
	
	For the upper bound, assume that we place a set with five unit-square guards~$g_1, g_2, g_3, g_4, g_5$ that can be shattered. 
	We denote viewpoints as $v_S$ with $S\subseteq\{1,2,3,4,5\}$. 
	Let $P_{i,j}, i\neq j\in\{1,\ldots,5\}$ denote the shortest path from the guard~$g_i$ to a guard $g_j$ along which the viewpoint $v_{\{i,j\}}$ is located. 
	In particular, $P_{i,j}$ includes the shortest paths from $g_i$ to $v_{\{i,j\}}$ and from $g_j$ to $v_{\{i,j\}}$, as this determines the rest budget for both guards at $v_{\{i,j\}}$.
	
	We start with four guards $g_1, \ldots, g_4$. 
	To generate all the ``pair'' viewpoints, $v_{\{i,j\}}$, $\{i,j\}\subseteq\{1,2,3,4\}$, we need to embed the graph $G_{4p}$, as shown in~\cref{fig:vc-holes-lb-d}, where each edge represents a path $P_{i,j}$ (the color of each guard reaches equally far into each edge, e.g., some of the paths reflected in these edges include wiggles); 
	\begin{figure}[htb]
		\centering
		\includegraphics[scale=0.275]{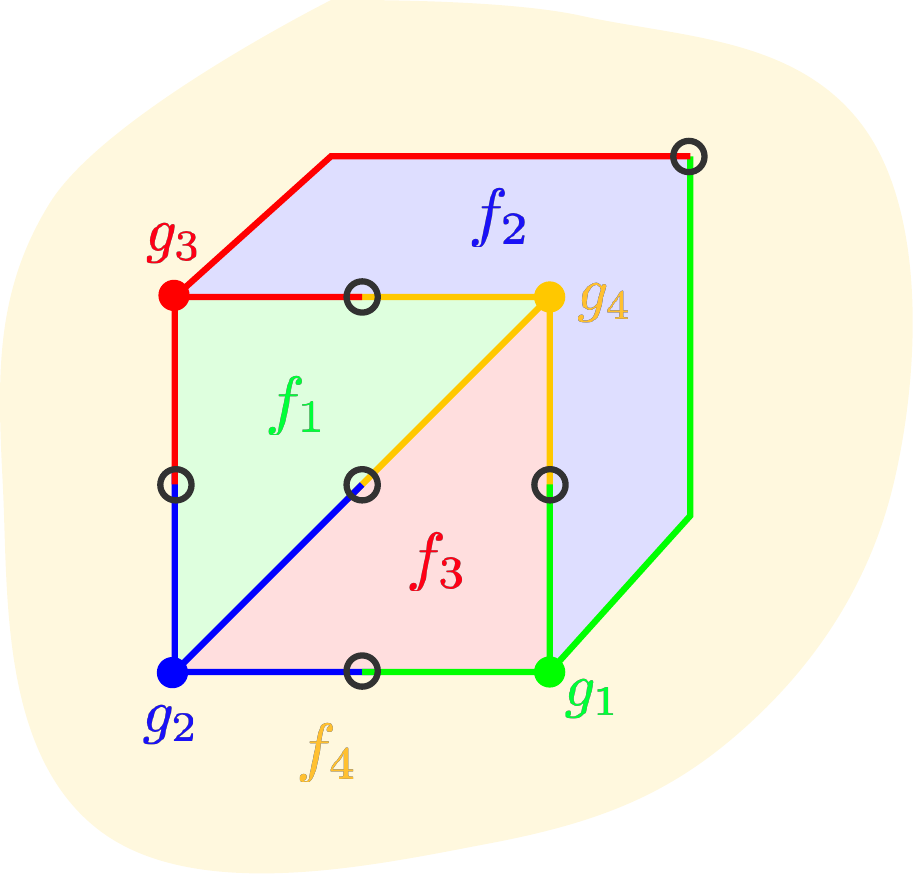}
		\caption{The graph~$G_{4p}$ with vertices $g_1, \ldots, g_4$ shown in green, blue, red, and yellow, respectively. Gray circles indicate ``pair''-viewpoints, and are not vertices of the graph.}
		\label{fig:vc-holes-lb-d}
	\end{figure}
	in particular, $G_{4p}$~contains edges representing the paths $P_{1,2},P_{1,3},P_{1,4},P_{2,3},P_{2,4}$, and $P_{3,4}$, and the face $f_i$ of $G_{4p}$ is incident to all guards $g_j, j\in\{1,2,3,4\}$ except for~$g_i$.
	Of course, in the resulting polyomino, the edges could be embedded in larger blocks of unit squares. 
	However, given the upper bound of $3$ on the VC dimension in simple polyominoes, by~\cref{th:vc-sim}, we know that at least one of the four faces $f_1, \ldots, f_4$ (and in fact one of $f_1, f_2, f_3$) of $G_{4p}$ must contain a hole. 
	A fifth guard $g_5$ must be located in one of the four faces. 
	Let this be face $f_i$. 
	As $g_i$ is not incident to $f_i$, the path $P_{i,5}$ from $g_5$ to $g_i$, must intersect at least one of the other paths represented by the edges in $G_{4p}$; let this be the path~$P_{j,\ell}$. 
	By~\cref{le:rest-budget}, one of the viewpoints $v_{\{i,5\}}$ and $v_{\{j,\ell\}}$ cannot be realized, as a guard from the other pair will always see such a viewpoint too; a~contradiction.
\end{proof}

\section{NP-completeness for 1-Thin Polyominoes with Holes}\label{sec:np}
In this section, we show that the decision version of the M$k$GP is \NP-complete, even in \num{1}-thin polyominoes with holes. 
However, as the dual graph of a \num{1}-thin polyomino without holes is a tree,  an optimal solution can be obtained in linear time~\cite{ack-ltamk-22,KM16}.

For showing \NP-hardness of our problem, we utilize \planarsat which de Berg and Khosravi~\cite{bk-obsps-12} proved to be \NP-complete, and that can be stated as follows.

\begin{problem}[\planarsat]
	Let $X=\{x_1, \ldots, x_n\}$ be a set of Boolean variables and $\varphi=C_1 \wedge C_2 \wedge\ldots \wedge C_m$ be a formula in conjunctive normal form defined over these variables, where each clause $C_i$ is the disjunction of at most three variables. 
	Let each clause be monotone, i.e., each clause consists of only negated or unnegated literals, and 
	let the bipartite variable-clause incidence graph be planar with a rectilinear embedding. 
	This constitutes a monotone, rectilinear representation of a \planarsat instance. 
	Decide whether the instance is satisfiable. 
\end{problem}

In the remainder of this section, we prove the following theorem.

\begin{theorem}\label{th:np}
	For every $k\geq 2$, the decision version of the M$k$GP is \NP-complete, even in~\num{1}-thin polyominoes with holes.
\end{theorem}

\begin{proof}
	Membership in \NP\@ follows easily, as we can verify in polynomial time for a given proposed solution whether each square of the respective polyomino is covered by the guards.
	Hence, it remains to show \NP-hardness. 
	
	Our reduction is from \planarsat.
	Given an instance $\varphi$ of \planarsat with incidence graph $G_\varphi$, we show how to turn a rectilinear, planar embedding of $G_\varphi$ into a polyomino~$P$, such that a solution to the M$k$GP in $P$ yields a solution to $\varphi$, thereby showing \NP-hardness. 
	At a high level, our reduction consists of four gadgets: variable gadgets to represent the variables of $\varphi$, split~gadgets to duplicate variable assignments, wire gadgets to connect variables to clauses, and clause gadgets to form the clauses of $\varphi$.
	
	\subparagraph{Variable Gadget} The \emph{variable gadget} is shown in~\cref{fig:np-var-a} for $k=2$: a four-pronged polyomino structure, where two prongs end in corridors to other gadgets, while the other two prongs are connected by a corridor with two unit-square niches.
	
	\begin{figure}[htb] 
		\centering
		\begin{subfigure}[b]{0.33\textwidth}
			\centering
			\includegraphics[scale=0.6]{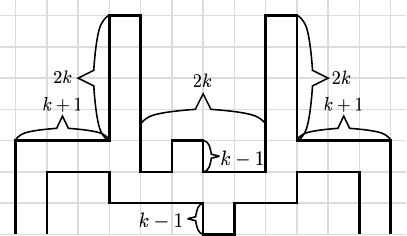}
			\caption{}
			\label{fig:np-var-a}
		\end{subfigure}%
		\begin{subfigure}[b]{0.33\textwidth}
			\centering
			\includegraphics[scale=0.6]{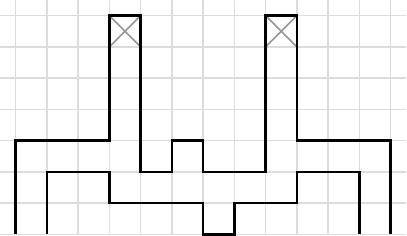}
			\caption{}
			\label{fig:np-var-b}
		\end{subfigure}%
		\begin{subfigure}[b]{0.33\textwidth}
			\centering
			\includegraphics[scale=0.6]{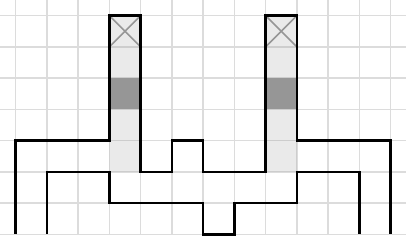}
			\caption{}
			\label{fig:np-var-c}
		\end{subfigure}\\\vspace*{0.2cm}
		\begin{subfigure}[b]{0.33\textwidth}
			\centering
			\includegraphics[scale=0.6]{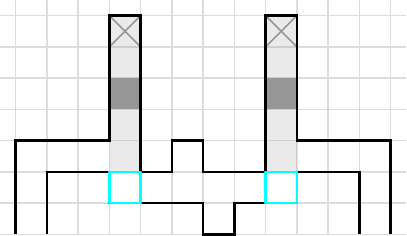}
			\caption{}
			\label{fig:np-var-d}
		\end{subfigure}%
		\begin{subfigure}[b]{0.33\textwidth}
			\centering
			\includegraphics[scale=0.6]{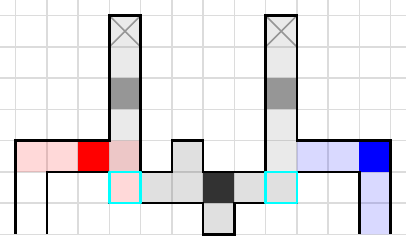}
			\caption{}
			\label{fig:np-var-e}
		\end{subfigure}%
		\begin{subfigure}[b]{0.33\textwidth}
			\centering
			\includegraphics[scale=0.6]{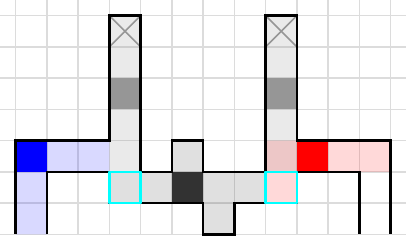}
			\caption{}
			\label{fig:np-var-f}
		\end{subfigure}%
		\caption{\label{fig:np-var} (a) Variable gadget for $k=2$. To cover the two marked squares in (b), the furthest into the gadget we can place guards are the two locations shown in (c). (d)~the two turquoise squares within distance  $2(k+1)$ are then still uncovered and a single guard cannot cover both of these. (e)~if a guard in the lower corridor covers the right turquoise square,  the visibility region of the right guard extends two squares into the right corridor, the red guard must cover the left turquoise square and cannot see into the left corridor. (f) shows a mirrored situation to that in (e).}
	\end{figure} 
	To cover the two marked squares in~\cref{fig:np-var-b}, the guard positions that cover the largest number of squares of the gadget are given in~\cref{fig:np-var-c}. 
	The $k$-hop-visibility regions of these guards leave a width-$(2k+2)$ corridor of uncovered squares at the bottom of the gadget. 
	The two extremal squares of that corridor,  highlighted in turquoise in~\cref{fig:np-var-d}, cannot be covered by a single guard. 
	Because of the two square niches attached to the width-$(2k+2)$ corridor, the only two positions for a potential guard that allow us to cover all of the variable gadget with five guards are vertically below and over these two niches. 
	If we pick the right of these positions, as in~\cref{fig:np-var-e}, we cover the right turquoise square, the last uncovered square on that side has distance $2$ to this turquoise square, and we can place a (blue) guard that sees distance $2$ into the right vertical corridor exiting the gadget.  
	The left turquoise square remains uncovered by the guard in the width-$(2k+2)$ corridor, hence, a guard placed in the left horizontal corridor that covers this square cannot extend into the left corridor exiting the gadget. 
	\cref{fig:np-var-f} shows the mirrored case. 
	We have exactly two sets of five guards placed within the variable gadget that cover the complete gadget (four guards are not sufficient), one refers to setting the variable to true (as shown in~\cref{fig:np-var-e}), the other to setting it to false (as in~\cref{fig:np-var-f}).
	
	\subparagraph{Wire and Split Gadget} The \emph{wire gadget} is simply an extended corridor from the two vertical corridor exits of a variable gadget.  
	A wire may be split using the \emph{split gadget} shown in~\cref{fig:np-split-a,fig:np-split-b}. 
	The split is located at a guard position of one truth setting (here the blue in~\cref{fig:np-split-a}), hence, that square is the last square covered by a guard representing the other assignment. 
	Given that our guards can look around an arbitrary number of corners (depending on $k$), we can bend our wires at any position.
	
	\begin{figure}[htb]
		\centering
		\begin{subfigure}[b]{.5\textwidth}
			\centering
			\includegraphics[scale=0.55]{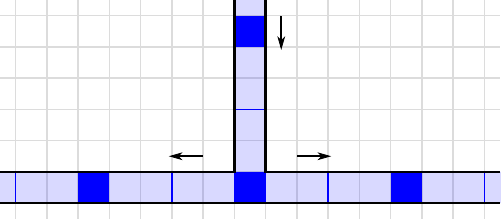}
			\caption{}
			\label{fig:np-split-a}
		\end{subfigure}%
		\begin{subfigure}[b]{.5\textwidth}
			\centering
			\includegraphics[scale=0.55]{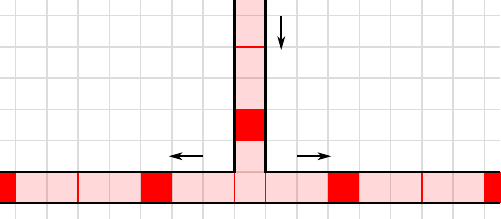}
			\caption{}
			\label{fig:np-split-b}
		\end{subfigure}%
		\caption{\label{fig:np-split} Split gadget for $k=2$ with the two possible truth assignments.}
	\end{figure}
	
	\subparagraph{Clause Gadget} 
	The \emph{clause gadget} is shown in~\cref{fig:np-clause-a} for $k=2$: a trident, with distance $k+1$ between its prongs.
	\begin{figure}[htb] 
		\centering
		\begin{subfigure}[b]{0.25\textwidth}
			\includegraphics[scale=0.55]{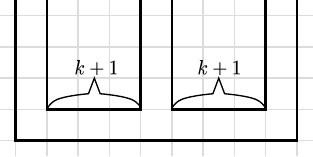}
			\caption{}
			\label{fig:np-clause-a}
		\end{subfigure}\hfill
		\begin{subfigure}[b]{0.25\textwidth}
			\includegraphics[scale=0.55]{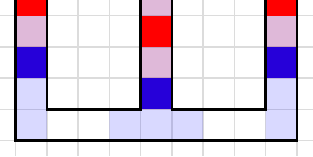}
			\caption{}
			\label{fig:np-clause-b}
		\end{subfigure}\hfill
		\begin{subfigure}[b]{0.25\textwidth}
			\includegraphics[scale=0.55]{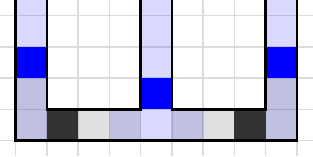}
			\caption{}
			\label{fig:np-clause-c}
		\end{subfigure}\hfill
		\begin{subfigure}[b]{0.25\textwidth}
			\includegraphics[scale=0.55]{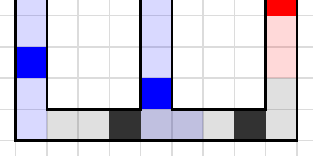}
			\caption{}
			\label{fig:np-clause-d}
		\end{subfigure}\hfill\\\vspace*{0.2cm}
		\begin{subfigure}[b]{0.25\textwidth}
			\includegraphics[scale=0.55]{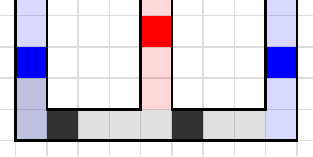}
			\caption{}
			\label{fig:np-clause-e}
		\end{subfigure}\hfill
		\begin{subfigure}[b]{0.25\textwidth}
			\includegraphics[scale=0.55]{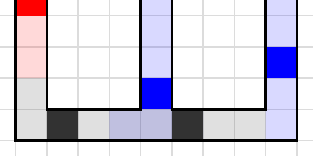}
			\caption{}
			\label{fig:np-clause-f}
		\end{subfigure}\hfill
		\begin{subfigure}[b]{0.25\textwidth}
			\includegraphics[scale=0.55]{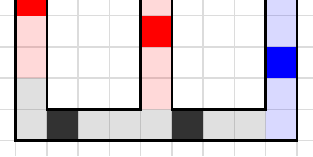}
			\caption{}
			\label{fig:np-clause-g}
		\end{subfigure}\hfill
		\begin{subfigure}[b]{0.25\textwidth}
			\includegraphics[scale=0.55]{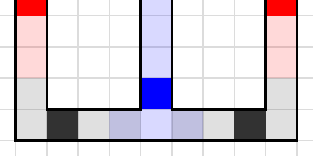}
			\caption{}
			\label{fig:np-clause-h}
		\end{subfigure}\hfill\\\vspace*{0.2cm}
		\begin{subfigure}[b]{0.25\textwidth}
		\end{subfigure}\hfill
		\begin{subfigure}[b]{0.25\textwidth}
			\includegraphics[scale=0.55]{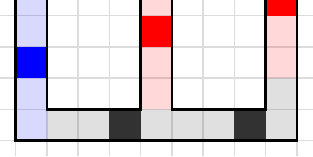}
			\caption{}
			\label{fig:np-clause-i}
		\end{subfigure}\hfill
		\begin{subfigure}[b]{0.25\textwidth}
			\includegraphics[scale=0.55]{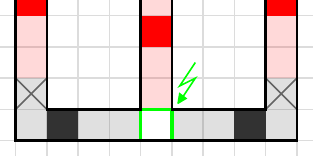}
			\caption{}
			\label{fig:np-clause-j}
		\end{subfigure}\hfill
		\begin{subfigure}[b]{0.25\textwidth}
		\end{subfigure}\hfill
		\caption{\label{fig:np-clause} (a) Clause gadget for $k=2$. (b) positions for all possible guards for truth assignments fulfilling the clause (blue, ``T'') and truth assignments not fulfilling the clause (red, ``F''). (c)--(j) show the different truth settings; in all but FFF, two additional guards within the clause are sufficient and necessary to cover the clause. (c)~TTT; (d)~TTF; (e)~TFT; (f) FTT; (g) FFT; (h) FTF; (i) TFF; (j) FFF.}
	\end{figure}
	Each prong is connected (by a wire) to a variable gadget. 
	In~\cref{fig:np-clause-b}, we depict all possible positions of guards in corridors that connect to a variable gadget with an assignment fulfilling the clause in blue, and all possible positions of guards in corridors that connect to a variable gadget with an assignment not fulfilling the clause in red. 
	The guards in the center prong are one square further into their corridor than those in the outer prongs.  
	For odd $k$ this yields an odd distance between guard positions in the center and in either of the outer prongs, which we cannot achieve automatically in a square grid. 
	However,  in a corridor, a guard's visibility region spans $2k+1$ squares,  which is odd for all $k$. 
	Hence, depending on whether we construct a corridor with a length that requires an even or an odd number of guards, we can cover an even or an odd distance.
	
	After having settled that we can achieve the odd distance between guard positions in the center prong and either of the outer prongs, we check the clause gadget: 
	As depicted in \cref{fig:np-clause-c,fig:np-clause-d,fig:np-clause-e,fig:np-clause-f,fig:np-clause-g,fig:np-clause-h,fig:np-clause-i}, in all but one cases of truth-value assignments, two additional guards in the clause gadget are necessary and sufficient to cover the gadget. 
	If the assignment does not fulfill the clause as shown in \cref{fig:np-clause-j}, we obtain a path of uncovered squares of length $(2\cdot(2k+1)+1)$, which is impossible to cover with two guards. 
	Hence, we need three guards for the clause gadget if and only if no variable has a truth setting fulfilling the clause.
	
	\subparagraph{Global Construction} We start from a planar, rectilinear embedding of $G_\varphi$ on an $O(n)\times O(n)$ grid~\cite{bk-bhogd-98}, scaled by a constant factor.  
	Then, we locally replace each clause by one clause gadget, each variable by one variable gadget,  and edges either by a single corridor from a clause gadget to a variable gadget (if the literal appears only in one clause), or by corridors with additional split gadgets.  
	Because we can steer a corridor's length to be even or odd by requiring an even or odd number of guards, we can always place all gadgets and connect them as desired. 
	This construction requires a guard set of size $\Gamma = 5V + 5C + \nicefrac{W}{(2k+1)}$, where $V$ and $C$ is the number of variables and clauses, respectively, and $W$ is the number of unit squares that make up all wire and split gadgets.
	
	\Cref{fig:np-ex} depicts an exemplary construction of the polyomino for the Boolean formula $\varphi=(\overline{x_1}\vee \overline{x_2} \vee \overline{x_3})\wedge (x_1\vee x_2 \vee x_4) \wedge(\overline{x_1}\vee\overline{x_3} \vee \overline{x_4})$ for $k=2$.
	
	\begin{figure}[htb]
		\centering
		\includegraphics[scale=0.37]{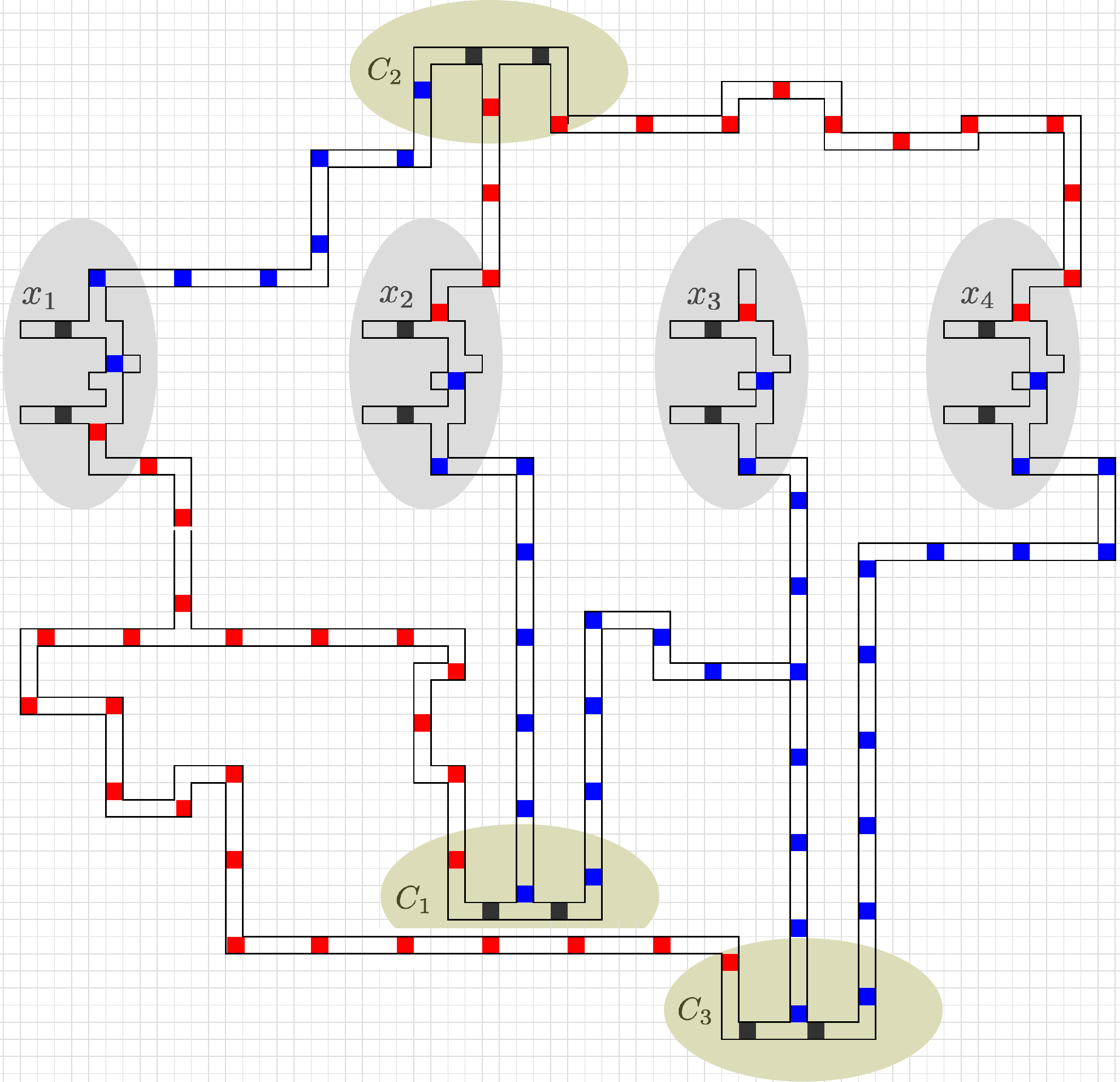}
		\caption{\label{fig:np-ex} The polyomino $P$ as an instance of the M$k$GP derived from the Boolean formula $\varphi=(\overline{x_1}\vee \overline{x_2} \vee \overline{x_3})\wedge (x_1\vee x_2 \vee x_4) \wedge(\overline{x_1}\vee\overline{x_3} \vee \overline{x_4})$, $k=2$.  We set $x_1=T$, and $x_2=x_3= x_4=F$.}
	\end{figure}
	
	\begin{claim}\label{claim:hardness1}
		If $\varphi$ is satisfiable, then there is a guard set of size $\Gamma$ under $k$-hop visibility.
	\end{claim}

	\begin{proof}\let\qed\relax
		Consider a satisfying assignment $\alpha$ for the Boolean formula $\varphi$, and the respective polyomino $P$ derived from $\varphi$ according to the described construction.
		For each variable in $\varphi$, we place five guards within the corresponding variable gadget in $P$ (where we place according to the truth assignment in~$\alpha$, see also~\cref{fig:np-var}).
		After placing these guards, there is a unique placement of guards within the wire and split gadgets (as every cell must be guarded). 
		As~we assume that $\varphi$ is satisfied by the assignment $\alpha$, placing additional five guards according to~\cref{fig:np-clause} suffices to cover~$P$.		
	\end{proof}

	\begin{claim}\label{claim:hardness2}
		If there is a guard set of size $\Gamma$ under $k$-hop visibility, then $\varphi$ is satisfiable.
	\end{claim}
	
	\begin{proof}\let\qed\relax
		Consider a guard set of size $\Gamma$ under $k$-hop visibility for the polyomino $P$.
		To cover all cells that belong to wire and split gadgets, we require at least $W$ many guards. 
		As described previously, every variable gadget requires at least five guards; this accumulates in a total of $5V$ guards, where $V$ again is the number of variables.
		So, we only have $5C$ guards left, where $C$ is the number of clauses.
		As we considered a feasible guard set, every clause needs at least five guards, and every clause can be covered with five guards only if parts of them are already covered by guards placed in incident wires, this induces a corresponding placement of guards within the variables gadgets.
		This placement provides a satisfying assignment for~$\varphi$.
	\end{proof}
	\Cref{claim:hardness1,claim:hardness2} complete the proof of~\cref{th:np}.
\end{proof}
\pagebreak

\section{A Linear Time 4-Approximation for Simple 2-Thin Polyominoes}\label{sec:appx}
As already mentioned, there exists a PTAS for $k$-hop domination in $H$-minor free graphs~\cite{fl-cetlt-21}. 
However, the exponent of $n$ in the runtime may be infeasible for realistic applications, where $n$ is extremely large. 
On the other hand, the exact algorithm for graphs with treewidth~$tw$ has running time $O((2k+1)^{tw}\cdot n)$~\cite{BL16}, which may be too large if $k=\Omega(n)$, even for small $tw$ (in fact, it is not hard to show that $2$-thin polyominoes have constant $tw$: $K_4$ is not a minor, hence, we have $tw=2$ for $2$-thin polyominoes).

Therefore, we present a linear-time $4$-approximation algorithm for the M$k$GP in simple $2$-thin polyominoes, for any value of $k\in \mathbb{N}$.
The runtime of our algorithm does not depend on $k$.
The overall idea is to construct a tree~$T$ on a polyomino $P$, and let $T$ lead us in placing guards in~$P$; this is inspired by the linear-time algorithm for trees by Abu-Affash et al.~\cite{ack-ltamk-22} for the equivalent problem of $k$-hop domination. 
In each iteration step, we place $1$, $2$, or $4$ guards and $1$~witness. 
Because the cardinality of a witness set is a lower bound on the cardinality of any guard set, this yields a $4$-approximation. 

We first provide a description of the involved tree, and then describe how we exploit it within our algorithm.

\subparagraph*{Skeleton Graph Construction}
Let $P$ be a simple $2$-thin polyomino.
A vertex $v$ of a cell~$s\in P$ is \emph{internal} if it does not lie on the boundary of~$P$. 
Because $P$ is $2$-thin, any unit square $s\in P$ can have at most $3$ internal vertices.
Let~$I$ be the set of internal vertices of unit squares in $P$.
For any $u,v\in I$, we add the edge~$\{u,v\}$ to~$E_I$ if one of the following holds (see~\cref{fig:4-appx-b}):
\begin{enumerate}
	\item $u$ and $v$ belong to the same cell and $\|u-v\|=1$. 
	\item $u$ and $v$ belong to two different cells that share an edge and both vertices of this edge are not internal.
	\item $u$ and $v$ belong to the same cell $s$ and both other vertices belonging to $s$ are not~internal.
\end{enumerate}
As $P$ is $2$-thin, we cannot create a four-cycle with edges stemming from Criterion~1. Furthermore, if we would create a cycle with edges stemming from Criterion 2 or~3, the polyomino would have a hole.
Hence, because $P$ is a simple $2$-thin polyomino, the edges of $E_I$ form a forest $T_I$ on $I$.

For each unit square $s\in P$ that does not have an internal vertex, place a point~$b_s$ in the center of $s$. 
We call $s$ a boundary square, $b_s$ a boundary node, and denote by~$B$ the set of all boundary nodes.
For any $b_s,b_{s'}\in B$ such that $s,s'\in P$ share an edge, we add the edge $\{b_s,b_{s'}\}$ to $E_B$. 
Notice that the edges of~$E_B$ form a forest $T_B$ on $B$.

We now connect $T_B$ and $T_I$ (see~\cref{fig:4-appx-c}). 
Let $s$ be a boundary square that shares an edge with a non-boundary square $s'$. Then $s'$ has at most two internal vertices.
If $s'$ has a single internal vertex $v$, we simply add $\{b_s,v\}$ to $E_{con}$. 
Else, $s'$ has two internal vertices $v,u$, and we add an artificial node $x_{v,u}$ to the set $X$, and the edges $\{b_s, x_{v,u}\}$,  $\{u, x_{v,u}\}$, and $\{v, x_{v,u}\}$ to $E_{con}$. We then remove the edge $\{u,v\}$ from $E_I$.

Let $T$ be the graph on the vertex set $V=I\cup B\cup X$ and the edge set ${E=E_I\cup E_B\cup E_{con}}$. 
Note that, as $P$ is simple, no cycles are created when connecting $T_I$ and $T_B$; thus, $T$ is a tree. 
Moreover, the maximum degree of a node in $T$ is $4$ (for some nodes in $X$ and $I$).

\subparagraph{Associated Squares} 
We associate with each node $v\in T$ a block $S(v)$ of unit squares from~$P$ as follows: 
\begin{enumerate}
	\item For $v=x_{u,v}\in X$, $S(v)$ consists of two unit squares with the edge $\{u,v\}$.
	\item For $v\in I$, $S(v)$ consists of a $2\times 2$ block of unit squares with internal vertex~$v$. 
	\item For $v=b_s\in B$, $S(v)=\{s\}$.
\end{enumerate}

\subparagraph{The Algorithm} As already mentioned, we basically follow the lines of the algorithm for $k$-hop dominating sets in trees~\cite{ack-ltamk-22}, with several important changes.

\smallskip
We start by picking an arbitrary node $r$ from $T$ as a root. For a node $u\in T$, denote by~$T_u$ the subtree of $T$ rooted at $u$. 
Any path between a unit square associated with a node in $T_u$ and a unit square associated with a node in $T\setminus T_u$ includes a unit square from~$S(u)$.
For every node $u\in T$, let $h(T_u)=\max_{v\in T_u,s\in S(v)}\min_{s'\in S(u)}d_P(s,s')$, where~$d_P(s,s')$ denotes the hop distance between the cells $s$ and $s'$ in $P$. 
In other words, $h(T_u)$ is the largest hop distance from a unit square in $\bigcup_{v\in T_u}S(v)$ to its closest unit square from $S(u)$. 

For each cell $s' \in \bigcup_{v\in T_u}S(v)$, the minimum distance $\min_{s\in S(u)}d_P(s,s')$ is assumed at a particular unit square $s$,  we denote by $M(s)$ the set of all these cells for which that distance is assumed for $s$, and set $h_s(T_u)=\max_{s'\in M(s)} d_P(s, s')$. 
Note that if $h(T_u)=k$, and we pick $S(u)$ for our guard set, then every unit square associated with a node in $T_u$ is guarded. 

Initialize an empty set $D$ (for the $k$-hop-visibility guard set), and compute $h(T_u)$ for every $u\in T$ and $h_s(T_u)$ for every $s\in S(u)$.
In addition, for every unit square~$s\in P$ set $\rb_D(s)=-1$ (up to a rest budget of $0$, $s$ is $k$-hop visible to the nodes in $D$). 
This parameter marks the maximum rest budget of the unit square~$s$ over all squares in the guard set $D$.
We~run a DFS algorithm starting from the root $r$, as follows; let $u$ be the current node in the DFS~call. Note that if $h(T_r)\le k$, then we can simply return $S(r)$ as our guard set.

\begin{enumerate}
	\item If $h(T_u)=k$, we add $S(u)$ to $D$, remove $T_u$ from $T$, and set $\rb_D(s)=k$ for every $s\in S(u)$.
	\item Else, if $k-1 \ge h(T_u)\geq k-2$ and $\min_{s\in S(p(u))} d(s,s')>k$ for the parent $p(u)$ of $u$ and $s'\in T_u$ being the unit square that realizes $h(T_u)$,  we add $S(u)$ to $D$, remove $T_u$ from $T$, and set $\rb_D(s)=k$ for every $s\in S(u)$.
	\item Else, for each child $v$ of $u$ with $h(T_v)\ge k-2$, we run the DFS algorithm on $v$. 
	Then we update $h(T_u)$ and $h_s(T_u)$, $\rb_D(s)$ for every $s\in S(u)$, according to the values calculated for all children of $u$. 
	\begin{enumerate}
	\item We check if the remaining $T_u$ is already guarded by $D$, by considering $h_s(T_u)$ and $\rb_D(s)$ for every $s\in S(u)$, where we only consider associated unit squares with negative rest budget.
	\item Else, if the new $h(T_u)$ is now exactly $k$ or if the condition from point 2 holds, then again we add $S(u)$ to $D$, remove $T_u$ from $T_r$, and set $\rb_D(s)=k$ for every $s\in S(u)$.
	\end{enumerate}
\end{enumerate}

If, at the end of the DFS run for $r$, we have $\rb_D(s)=-1$ for some $s\in S(r)$, then we add $S(r)$ to $D$. 
We give an example of our algorithm in~\cref{fig:4-appx}.

\begin{figure}[htb] 
	\centering
	\begin{subfigure}[b]{0.5\textwidth}%
		\centering
		\includegraphics[scale=0.28]{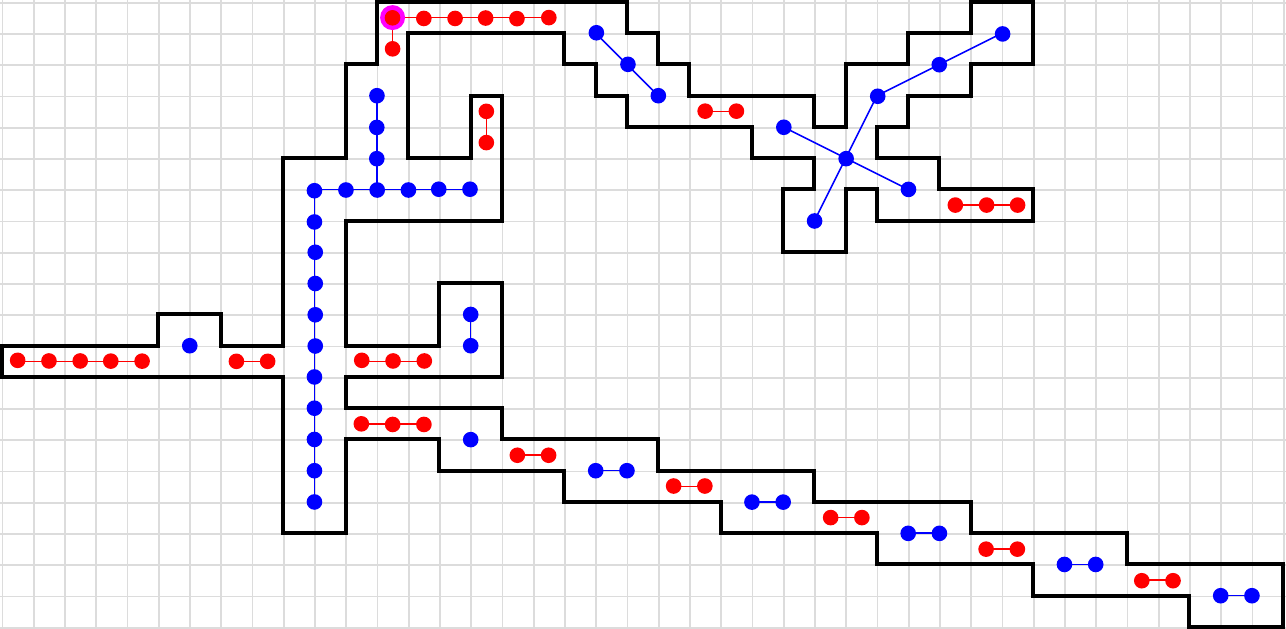}
		\caption{}
		\label{fig:4-appx-b}
	\end{subfigure}%
	\begin{subfigure}[b]{0.5\textwidth}%
		\centering
		\includegraphics[scale=0.28]{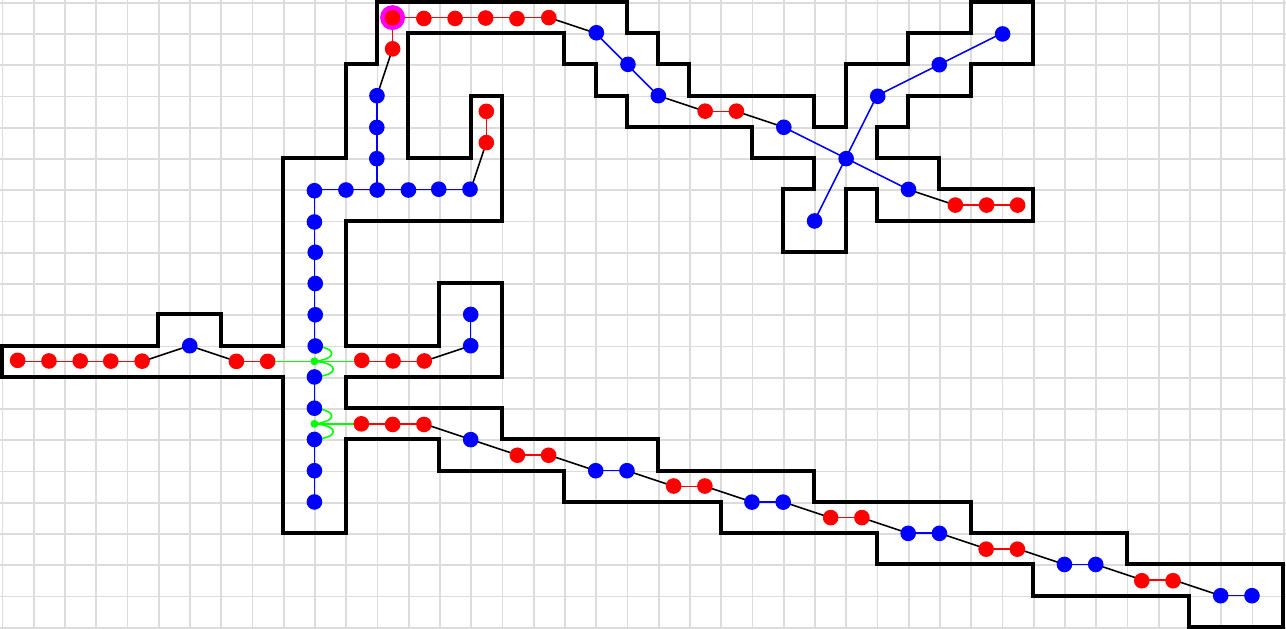}
		\caption{}
		\label{fig:4-appx-c}
	\end{subfigure}\\\vspace*{0.3cm}
	\begin{subfigure}[b]{0.5\textwidth}%
		\centering
		\includegraphics[scale=0.28]{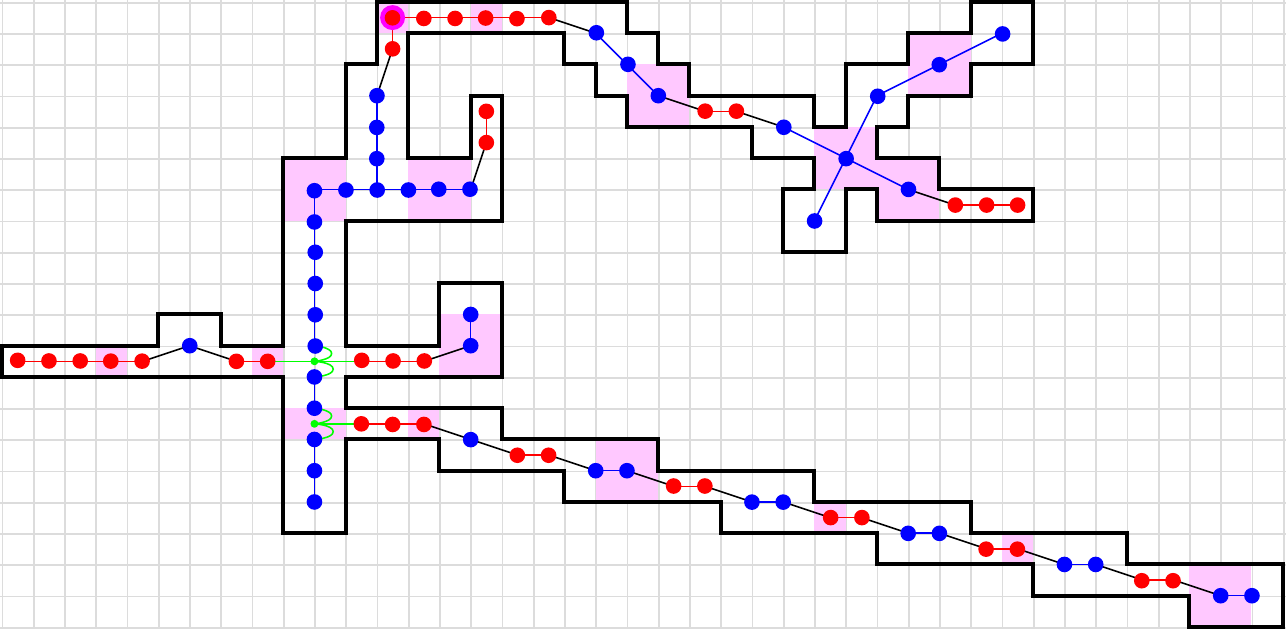}
		\caption{}
		\label{fig:4-appx-d}
	\end{subfigure}%
	\begin{subfigure}[b]{0.5\textwidth}%
		\centering
		\includegraphics[scale=0.28]{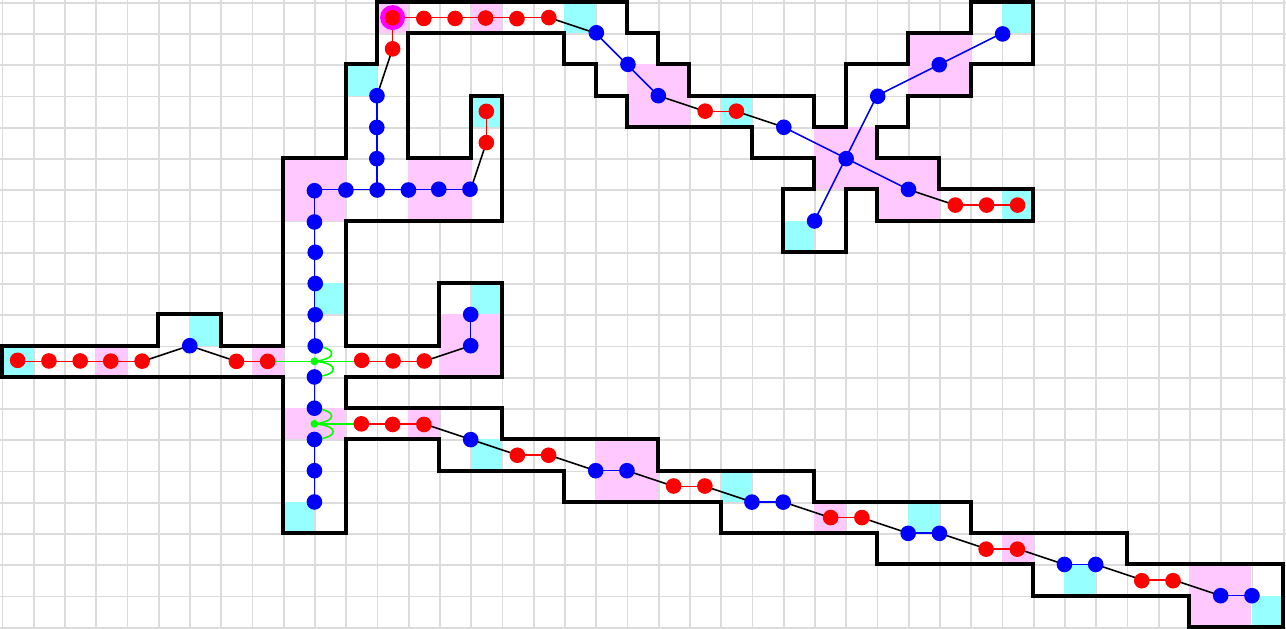}
		\caption{}
		\label{fig:4-appx-e}
	\end{subfigure}%
	\caption{\label{fig:4-appx} 
		Example for our algorithm for $k=3$: 
		(a) Polyomino $P$, in black, with the associated vertices: vertices in \textcolor{blue}{$I$} in blue (\textcolor{blue}{$\bullet$}), vertices in \textcolor{red}{$B$} in red (\textcolor{red}{$\bullet$}), the root~$r$ is indicated in magenta (\textcolor{magenta}{$\sbullet[2]$}), the trees $T_I$ and $T_B$ with edges in $E_I$ are shown in blue, edges in $E_B$ are shown in red; 
		(b) connecting $T_I$ and~$T_B$, vertices~$x_{v,u}$ and their incident edges are shown in green (\textcolor{green}{$\bullet$}), all other connecting edges are shown in black; 
		(c) unit squares added to $D$ in light pink; 
		(d) placement of witnesses from the proof of~\cref{th:4-appx} in turquoise.
	}
\end{figure} 

\medskip
We show that, after termination of the described algorithm, $D$ is a $k$-hop-visibility guard set for the given simple $2$-thin polyomino $P$ of size at most $4\cdot \textrm{OPT}$ for all $k\in \mathbb{N}$, where $\textrm{OPT}$ is the size of an optimal solution.

\begin{theorem}\label{th:4-appx}
	There is a linear-time algorithm computing a $4$-approximation for the M$k$GP in simple \num{2}-thin polyominoes.
\end{theorem}
\begin{proof}
	During the algorithm we have just described, we remove a node $v$ from $T$ only if $S(v)$ is covered by cells in~$D$. 
	Because $\cup_{v\in T}S(v)=P$, $D$ is a $k$-hop-visibility guard set for $P$.
	
	Next, let $u_1,\dots, u_\ell$ be the sequence of nodes of $T$ such $S(u_i)$ was added to the set~$D$ during the algorithm. We show that in each $T_{u_i}$ we can find a witness unit square~$s_i$, such that no two witness unit squares $s_i\neq s_j$ have a single unit square in~$P$ within hop distance $k$ from both $s_i$ and $s_j$. This means that any optimal solution has size at least $\ell$ ($W=\{s_1,\ldots,s_{\ell}\}$ is a witness set with $|W|=\ell$). Because we add at most $4$~unit squares to~$D$ in each step of the algorithm, we get a solution of size at most $4\ell$, as required.
	
	\smallskip
	We choose $s_i$ to be the unit square from $T_{u_i}$ with maximum distance to its closest unit square from~$S(u_i)$, i.e., the unit square that realizes $h(T_{u_i})$. 
	We claim that there is no cell in~$P$ within hop distance $k$ from both $s_i$ and $s_j$ for any $j<i$.
	
	\smallskip
	If $S(u_i)$ was added to~$D$ because $T_{u_i}=k$, we had $s_i$ being the node realizing~$T_{u_i}$. 
	Hence, we have~$\rb_{D\setminus\{u_i\}}(s_i)=-1$, and thus, the distance from $s_i$ to any $s_j, j<i$ is at least $2k+1$.
	
	\smallskip
	If $S(u_i)$ was added to~$D$ because $h(T_{u_i})\geq k-2$ and $\min_{s\in S(p(u))} d(s,s_i)>k$ for the parent $p(u)$ of $u$ and $s_i\in T_{u_i}$ being the unit square that realizes $h(T_{u_i})$,  we know (because each unit square of the polyomino is an associated unit square of at least one node) that there is a unit square $s''\in S(u_i)$ with $d(s'',s_i)=k$.   Thus,  any witness placed after $s_i$ has distance to it of at least $2k+1$. Moreover, $\rb_{D\setminus\{u_i\}}(s_i)=-1$ and, thus,  $s_i$'s distance to any $s_j, j<i$ is at least $2k+1$.
	
	\smallskip
	As we initialize $h(T_u)$ for every $u\in T$ and $h_s(T_u)$ for every $s\in S(u)$ with a BFS-call, and update the values at most once for each square in linear time, the overall running time is $O(n)$.
\end{proof}

\bibliography{lit}
\end{document}